\documentclass[preprint,12pt]{elsarticle}

\usepackage{amssymb}
\usepackage{amsmath}
\newtheorem{lemma}{Lemma}
\newtheorem{theorem}{Theorem}
\newtheorem{corollary}{Corollary}
\newproof{proof}{Proof}
\newproof{aproof}{Proof of Lemma \ref{Lemm_Relationship}}

\journal{Finite Fields and Their Applications}

\begin{document}

\begin{frontmatter}



\title{Differential Spectrum of Some Power Functions With Low Differential Uniformity}


\author[SNU]{Sung-Tai Choi}
\author[SNU]{Seokbeom Hong}
\author[SNU]{Jong-Seon No}
\author[HU]{Habong Chung}
 %\sep Seokbeom Hong \sep \\Jong-Seon No \sep Habong Chung}

\address[SNU]{ Department of Electrical Engineering and
Computer Science, INMC\\ Seoul National University, Seoul 151-744, Korea}

\address[HU]{School of Electronics and Electrical Engineering\\ 
           Hongik University, Seoul 121-791, Korea}

\begin{abstract}
In this paper, for an odd prime $p$, the differential spectrum of the power function $x^{\frac{p^k+1}{2}}$ in $\mathbb{F}_{p^n}$ is calculated. For an odd prime $p$ such that $p\equiv 3\bmod 4$ and odd $n$ with $k|n$, the differential spectrum of the power function $x^{\frac{p^n+1}{p^k+1}+\frac{p^n-1}{2}}$ in $\mathbb{F}_{p^n}$ is also derived. From their differential spectrums, the differential uniformities of these two power functions are determined. We also find some new power functions having low differential uniformity.

\end{abstract}



\end{frontmatter}

% \linenumbers




%% References

\section{Introduction}

Let $p$ be a prime number and $\mathbb{F}_{p^n}$ the finite field with $p^n$ elements. Let $f(x)$ be  a mapping from $\mathbb{F}_{p^n}$ to $\mathbb{F}_{p^n}$. Let $N(a,b)$ denote the number of solutions $x\in \mathbb{F}_{p^n}$ of $f(x+a)-f(x)=b$, where $a\in \mathbb{F}_{p^n}^*$ and $b\in \mathbb{F}_{p^n}$. Then the differential uniformity $\Delta_f$ is defined as
\begin{align*}
	\Delta_f=\max_{a\in \mathbb{F}_{p^n}^*,b\in \mathbb{F}_{p^n}} N(a,b).
\end{align*} 
Nyberg \cite{Nyberg} defined a mapping to be differential $k$-uniform if $\Delta_f=k$. This differential uniformity is of interest in cryptography because differential and linear cryptanalysis exploit the weakness in the uniformity of the substitution functions which are used in data encryption standard (DES), advanced encryption standard (AES), and many other block cipher systems. For applications in cryptography, one would prefer functions having $\Delta_f$ as small as possible. Hence the functions with low $\Delta_f$ have been searched extensively \cite{Helleseth}--\cite{Edel}. 
 Especially for an odd prime $p$, there exist functions with $\Delta_f=1$, which are said to be perfect nonlinear (PN). The functions with $\Delta_f=2$ are said to be almost perfect nonlinear (APN). Some more functions having low differential uniformity are studied in \cite{Helleseth2} and  \cite{Bracken2}.

 Let $f(x)$ be the power function given as $f(x)=x^d$. For any $a\in \mathbb{F}_{p^n}^*$ and $b\in \mathbb{F}_{p^n}$, the differential equation $f(x+a)-f(x)=b$ can be rewritten as 
 \begin{align*}
 	a^d\Big(\big(\frac{x}{a}+1\big)^d-\big(\frac{x}{a}\big)^d\Big)=b,
 \end{align*}
 which means that
 \begin{align*}
 	N(a,b)=N(1,\frac{b}{a^d}).
 \end{align*}
 Hence, in dealing with power functions, we can only consider $N(1,b)$ instead of $N(a,b)$.

 The differential spectrum of the function $f(x)$ with $\Delta_f=k$ is defined as $(\omega_0,\omega_1,\dots,\omega_k)$, where $\omega_i$ denotes the number of $b\in \mathbb{F}_{p^n}$ such that $N(1,b)=i$.  In \cite{Blondeau2}, the differential spectrum of substitution functions is introduced and its relation to differential attacks on block ciphers is discussed.  In \cite{Blondeau}, the relationship between the differential spectrum of $x^{2^t-1}$ and $x^{2^{n-t+1}-1}$ in $\mathbb{F}_{2^n}$ is derived and the differential spectrum of $x^{2^t-1}$ for $t\in \{3,\lfloor n/2 \rfloor,\lceil n/2 \rceil+1,n-2\}$ is also calculated. Still, there have been not so many researches on the differential spectrum of certain functions. 
 
 In \cite{Helleseth}, for an odd prime $p$, the power function $x^{\frac{p^k+1}{2}}$ in $\mathbb{F}_{p^n}$ was first analyzed with respect to differential uniformity. It was shown that its differential uniformity is upper bounded as $\Delta_f\leq \gcd((p^k-1)/2,p^{2n}-1)$. Nevertheless, the upper bound is not tight in some cases of $p$, $n$, and $k$, which motivates us to derive the exact value of $\Delta_f$ for $x^{(p^k+1)/2}$ in this paper.

In this paper, for an odd prime $p$, the differential spectrum of $x^{\frac{p^k+1}{2}}$ in $\mathbb{F}_{p^n}$ is derived. For an odd prime $p$ such that $p\equiv 3\bmod 4$, odd $n$, and $k|n$, the differential spectrum of  $x^{\frac{p^n+1}{p^k+1}+\frac{p^n-1}{2}}$ in $\mathbb{F}_{p^n}$ is also derived. Based on the results, some new functions with low differential uniformity $\Delta_f$ are found. 
 
 This paper is organized as follows. In Section \ref{Pre}, some preliminaries and notations are stated. In Section \ref{Dif_Fun1}, the differential spectrum of  $x^{\frac{p^k+1}{2}}$ in $\mathbb{F}_{p^n}$ is proved. In Section \ref{Dif_Func2},  the differential spectrum of  $x^{\frac{p^n+1}{p^k+1}+\frac{p^n-1}{2}}$ in $\mathbb{F}_{p^n}$ is calculated. The conclusion is given in Section \ref{Concl}.

\section{Preliminaries and Notations}
\label{Pre}

Let $p$ be an odd prime, $\alpha$ be a primitive element of the finite field $\mathbb{F}_{p^n}$, and $p^n=sl+1$.  
 Then the cyclotomic classes $C_i,~0\leq i \leq s-1$, in $\mathbb{F}_{p^n}$ are defined as
 \begin{align*}
 	C_i=\{\alpha^{st+i}|~t=0,1,\dots,l-1\},~0\leq i\leq s-1.
 \end{align*}
 Note that $C_i$'s are pairwise disjoint and their union is the multiplicative group of $\mathbb{F}_{p^n}$ denoted by $\mathbb{F}_{p^n}^*=\mathbb{F}_{p^n}\setminus \{0\}$. Then the cyclotomic number $(i,j)_s$ is defined as the number of solutions $(x_i,x_j
 )\in C_i\times C_j$ for $x_i+1=x_j$.

 \begin{lemma}[Lemma 6 \cite{Storer}]\label{Lem_Cyclotomy_l_2}
 	When $s=2$, the cyclotomic numbers $(i,j)_2\equiv (i,j)$ are given as:\\
 	\noindent 1) $p^n\equiv 1\bmod 4$;
 	 \begin{align*}
 	 	(0,0)=\frac{p^n-5}{4};~(0,1)=(1,0)=(1,1)=\frac{p^n-1}{4}.
 	 \end{align*}
 	 \noindent 2) $p^n\equiv 3\bmod 4$;
 	\begin{align*}
 			(0,0)=(1,0)=(1,1)=\frac{p^n-3}{4};~(0,1)=\frac{p^n+1}{4}.
 	\end{align*} \hfill$\Box$
 \end{lemma}
 
 For $s=2$, let $E_{ij}$, $0\leq i,j\leq 1$, be the set defined as
 \begin{align}\label{Def_Cyc}
 	E_{ij}=\{x\in \mathbb{F}_{p^n}^*|x\in C_i~{\rm and~}x+1\in C_j\}.
 \end{align}
Then $(i,j)=|E_{ij}|$.

In the following lemma, we are going to express each $x\in E_{ij}$ in terms of the primitive element of $\mathbb{F}_{p^n}$ or $\mathbb{F}_{p^{2n}}$. Let $[a,b]$ denote the set of consecutive integers between $a$ and $b$ including $a$ and $b$, that is, $[a,b]=\{a,a+1,\dots,b\}$. 

 \begin{lemma}\label{Lem_Cyc}
Any element $x$ in $E_{00}$ can be represented as 
 \begin{align}\label{Eq_C00}
 	x=\Big(\frac{\alpha^t-\alpha^{-t}}{2}\Big)^2
 \end{align}
 where $t$ varies over $\mathcal{T}_1=[1,(p^n-3)/4]$ for $p^n\equiv 3\bmod 4$ and  over $\mathcal{T}_2=[1,(p^n-5)/4]$ for $p^n\equiv 1\bmod 4$. Any element $x$ in $E_{11}$ can be represented as
 \begin{align}\label{Eq_C11}
 	x=\gamma\Big(\frac{\alpha^t-\gamma^{-1}\alpha^{-t}}{2}\Big)^2
 \end{align}
 where $\gamma=-1$ and $t$ varies over $\mathcal{T}_1$ for $p^n\equiv 3\bmod 4$ and  $\gamma=-\alpha$ and $t$ varies over $\mathcal{T}_2\cup \{0\}$ for $p^n\equiv 1\bmod 4$. Any element $x$ in $E_{10}$ can be represented as 
 \begin{align}\label{Eq_C10}
 	x=\Big(\frac{\delta^{2t}-\delta^{-2t}}{2}\Big)^2
 \end{align}
 where $\delta=\beta^{(p^n-1)/2}$ and $\beta$ is a primitive element in $\mathbb{F}_{p^{2n}}$ and $t$ varies over $\mathcal{T}_1$ for $p^n\equiv 3\bmod 4$ and over $\mathcal{T}_2\cup \{(p^n-1)/4\}=[1,(p^n-1)/4]$ for $p^n\equiv 1\bmod 4$. Finally, any element $x$ in $E_{01}$ can be represented as 
 \begin{align}\label{Eq_C01}
 	x=\Big(\frac{\delta^{2t+1}-\delta^{-(2t+1)}}{2}\Big)^2
 \end{align}
 where $t$ varies over $\mathcal{T}_1\cup \{0\}=[0,(p^n-3)/4]$ for $p^n\equiv 3\bmod 4$ and over $\mathcal{T}_2\cup \{0\}=[0,(p^n-5)/4]$ for $p^n\equiv 1\bmod 4$. 
 \end{lemma}
 \begin{proof}
 		For $x\in E_{00}$, we can set $x+1=u^2$ and $x=v^2$ for some $u,v\in \mathbb{F}_{p^n}^*$. Then we have $u^2-v^2=(u+v)(u-v)=1$. Let $u+v=\alpha^t$. Then we have $u=(\alpha^t+\alpha^{-t})/2$ and $v=(\alpha^t-\alpha^{-t})/2$. Hence $x$ in $E_{00}$ is represented as $x=(\alpha^t-\alpha^{-t})^2/4$. Then we have to determine the range over which $t$ varies. From Lemma \ref{Lem_Cyclotomy_l_2}, we know that $|E_{00}|=(p^n-3)/4$ for $p^n\equiv 3\bmod 4$ and $(p^n-5)/4$ for $p^n\equiv 1\bmod 4$. It is easy to check that $\{\alpha^t,\alpha^{-t},-\alpha^t,-\alpha^{-t}\}$ induce the same $x$ in (\ref{Eq_C00}). Note that $t=0$ makes $x=0$ and $t=(p^n-1)/4$ makes $x=0$ when $p^n\equiv 1\bmod 4$. Hence $t$ varies over $1\leq t\leq (p^n-3)/4$ for $p^n\equiv 3\bmod 4$ and $1\leq t\leq (p^n-5)/4$ for $p^n\equiv 1\bmod 4$.
 	
 For $x\in E_{11}$, we can set $x+1=\gamma u^2$ and $x=\gamma v^2$ for some $u,v\in \mathbb{F}_{p^n}^*$, where $\gamma$ is a nonsquare in $\mathbb{F}_{p^n}^*$. Then we have $u^2-v^2=(u+v)(u-v)=\gamma^{-1}$. Let $u+v=\alpha^t$. Then we have $u-v=\gamma^{-1}\alpha^{-t}$ and thus $u=(\alpha^t+\gamma^{-1}\alpha^{-t})/2$ and $v=(\alpha^t-\gamma^{-1}\alpha^{-t})/2$. Hence $x\in E_{11}$ is represented as $x=\gamma(\alpha^t-\gamma^{-1}\alpha^{-t})^2/4$. Now, we have to determine the range over which $t$ varies. From Lemma \ref{Lem_Cyclotomy_l_2}, we know that $|E_{11}|=(p^n-3)/4$ for $p^n\equiv 3\bmod 4$ and $(p^n-1)/4$ for $p^n\equiv 1\bmod 4$. It is easy to check that $\{\alpha^t,-\alpha^t,\gamma^{-1}\alpha^{-t},-\gamma^{-1}\alpha^{-t}\}$ induce the same $x$ in (\ref{Eq_C11}). Clearly, for the case of $p^n\equiv 3\bmod 4$, if we set $\gamma=-1$, then each $t$ in $\mathcal{T}_1$ makes distinct $x$ in $E_{11}$. For the case of $p^n\equiv 1\bmod 4$, each $t$ in $\mathcal{T}_2$ makes distinct $x$ in $E_{11}$ for $\gamma=-\alpha$ similarly. 

For $x\in E_{10}$ or $E_{01}$, the proof becomes a little more tricky. For $x\in E_{10}$, we can set $x+1=u^2$ and $x=\gamma v^2$ for some $u,v\in \mathbb{F}_{p^n}$, where $\gamma$ is a nonsquare in $\mathbb{F}_{p^n}^*$. Then we have $u^2-\gamma v^2=1$, which can be factorized in $\mathbb{F}_{p^{2n}}$ as $u^2-\gamma v^2=(u+\lambda v)(u-\lambda v)=(u+\lambda v)(u+\lambda^{p^n}v)=(u+\lambda v)^{p^n+1}=1$, where $\lambda$ and $-\lambda=\lambda^{p^n}$ are the two solutions in $\mathbb{F}_{p^{2n}}$ of $X^2=\gamma$ \cite{Dickson}. Since $u+\lambda v$ is the $(p^n+1)$-st root of unity in $\mathbb{F}_{p^{2n}}$, we can set $u+\lambda v=\beta^{(p^n-1)t}=\delta^{2t}$, where $\delta=\beta^{(p^n-1)/2}$ and $\beta$ is a primitive element of $\mathbb{F}_{p^{2n}}$. 
Since $u+\lambda v=\delta^{2t}$ and $u-\lambda v=\delta^{-2t}$, we have $x=(\delta^{2t}-\delta^{-2t})^2/4$. Then we have to determine the range over which $t$ varies. From Lemma \ref{Lem_Cyclotomy_l_2}, we know that $|E_{10}|=(p^n-3)/4$ for $p^n\equiv 3\bmod 4$ and $(p^n-1)/4$ for $p^n\equiv 1\bmod 4$. Note that $\{\delta^{2t},\delta^{-2t},-\delta^{2t},-\delta^{-2t}\}$ induce the same $x$ in (\ref{Eq_C10}). The values $t=0$ and $t=(p^{n}+1)/2$ which make $x=0$ and $t=(p^{n}+1)/4$ which makes $x=-1$ should be excluded. Then each $t\in \mathcal{T}_1$ gives distinct $x$ for $p^n\equiv 3\bmod 4$ and so does $t\in \mathcal{T}_2\cup \{(p^n-1)/4\}$ for $p^n\equiv 1\bmod 4$. We can prove the case for $x\in E_{01}$ similarly. \hfill $\Box$
 \end{proof}

\section{The Differential Spectrum of $x^{\frac{p^k+1}{2}}$ in $\mathbb{F}_{p^n}$}
\label{Dif_Fun1}

In \cite{Helleseth}, for an odd prime $p$, the upper bound on differential uniformity $\Delta_f$ of the power function $f(x)=x^{\frac{p^k+1}{2}}$ in $\mathbb{F}_{p^n}$ is derived. The result is stated as in the following theorem.

\begin{theorem}[Theorem 11 \cite{Helleseth}]
\label{Thm_Helleseth}
	Let $f(x)=x^d$ be the function defined on $\mathbb{F}_{p^n}$, where $p$ is an odd prime and $d=(p^k+1)/2$. Then we have
	\begin{align*}
		\Delta_f\leq \gcd \Big(\frac{p^k-1}{2},p^{2n}-1\Big).
	\end{align*} \hfill$\Box$
\end{theorem}

However, in some cases of $p$, $n$, and $k$, the upper bound is not tight, which motivates us to derive the differential spectrum and the differential uniformity $\Delta_f$. The following lemmas are needed for the proof of the subsequent  lemmas and theorem. 

\begin{lemma}\label{Lem_Remainder}
	Define the set $\mathcal{A}=[1,N]$ for a positive integer $N$. Assume that $N\equiv r\bmod v$  for a nonzero integer $v$ and $q$ is a quotient so that $N=qv+r$. Let $n_{\mu}$ denote the number of elements $a\in \mathcal{A}$ such that $a\bmod v$ is either $+\mu$ or $-\mu$ for an integer $0\leq \mu\leq v/2$. Then $n_{\mu}$ is computed as:

\noindent 1) When $\mu=0$ or ${v/2}$ for even $v$;
	\begin{align*}
		n_{\mu}=
		\begin{cases}
		q+1,~&{\rm for~}\mu=\frac{v}{2}\leq r{~\rm with~even~}v\\
		q,~&{\rm for~}\mu=0~{\rm or~}\mu=\frac{v}{2}> r{~\rm with~even~}v.
		\end{cases}
	\end{align*}
	
\noindent 2) When $0<\mu<v/2$;		
	\begin{align*}
		n_{\mu}=
		\begin{cases}
		2(q+1),~&{\rm for~}v-r\leq \mu\leq r\\
		2q+1,~&{\rm for~}\mu\geq \max(v-r,r+1) {~\rm or~}\mu\leq \min(r,v-r-1)\\
		2q,~&{\rm for~}r<\mu<v-r.
		\end{cases}
	\end{align*} \hfill$\Box$
\end{lemma}
We will omit the proof because it is nothing more than a simple counting.

\begin{lemma}\label{Lem_gcd}
	Let $p$ be an odd prime and $l=\gcd(a,b)$. Let $a'=a/l$ and $b'=b/l$. Then
	\begin{align*}
		\gcd(p^a+1,p^b-1)=
		\begin{cases}
			p^l+1,&~{\rm for~odd~}a'~{\rm and~even}~b'\\
			2,&~{\rm otherwise}.
		\end{cases}
	\end{align*}
\end{lemma}
\begin{proof}
	Let $m=\gcd(p^a+1,p^b-1)$. Now, $p^l\equiv \pm 1\bmod m$ will be proved. By B\'ezout's identity, $l$ can be expressed as $l=ax+by$, where $x$ and $y$ are some integers.  Then we have
\begin{align}\label{Eq_Bezout}
	p^l\equiv p^{ax+by}\equiv (p^a)^x(p^b)^y\equiv (-1)^x(1)^y \equiv \pm 1 \bmod m,
\end{align}
which means that $m|p^l+1$ or $m|p^l-1$. Now, $m$ will be determined in the following three cases:

\vspace{0.1cm}
\noindent {\bf Case 1)} $2l|a$; 

From (\ref{Eq_Bezout}), we have $m|p^{2l}-1$. Since $m|p^{2l}-1$ and $a'$ is even, we have $m|p^a-1$. Since $m|p^a+1$, we have $m|((p^a+1)-(p^a-1))$, i.e., $m|2$. Since $m\geq 2$, we have $m=2$. 

\noindent {\bf Case 2)} For odd $a'$ and even $b'$; 

Since $a'$ is odd and $b'$ is even, we have $p^l+1|p^a+1$ and $p^l+1|p^b-1$. Hence we have $p^l+1|m$.  From $p^l+1|m$ and (\ref{Eq_Bezout}), we have $m=p^l+1$.

\noindent {\bf Case 3)} For odd $a'$ and odd $b'$; 

Assume that $m|p^l+1$. Since $m|p^l+1$ and $b'$ is odd, we have $m|p^b+1$. Since $m|p^b-1$, we have $m=2$. Now, assume that $m|p^l-1$. Since $m|p^a-1$ and $m|p^a+1$, we have $m=2$.  \hfill $\Box$
\end{proof}

Let $D_f(x)=f(x+1)-f(x)$ and $\mathcal{I}_{ij}$ be the image of $E_{ij}$ under $D_f$, that is,
\begin{align*}
  \mathcal{I}_{ij}=\{D_f(x)|x\in E_{ij}\}  
\end{align*}
where $i,j\in \{0,1\}$. Also, define the set $\mathcal{U}_{ij}(b)$, $b\in \mathcal{I}_{ij}$, as the set of elements $x\in E_{ij}$ such that $D_f(x)=b$. Let $\theta:t\mapsto x$ be the bijective mapping from $t$ to $x$ given in Lemma \ref{Lem_Cyc}. In the following Lemmas \ref{Lem_MapProperty_00}--\ref{E01_E10}, the cardinalities of each $\mathcal{I}_{ij}$ and $\mathcal{U}_{ij}(b)$'s, $b\in \mathcal{I}_{ij}$, will be determined. Let $e=\gcd(n,k)$ and $g=\gcd(2n,k)$ in the remainder of this section.

\begin{lemma}
\label{Lem_MapProperty_00}
For $\mathcal{I}_{00}$ and $D_f|_{E_{00}}$, we have

\noindent 1) For an odd $n/e$;
\begin{itemize}
	\item[] $|\mathcal{I}_{00}|=(p^n+p^e-2)/(2(p^e-1))$
	\item[] $|\mathcal{U}_{00}(b)|=\begin{cases}\frac{p^e-3}{4},&~{\rm for~}b=1{~\rm and~}p^n\equiv 3\bmod 4\\ \frac{p^e-5}{4},&~{\rm for~}b=1{~\rm and~}p^n\equiv 1\bmod 4 \\ \frac{p^e-1}{2},&~{\rm for~}b(\neq 1)\in \mathcal{I}_{00}.\end{cases}$
\end{itemize}

\noindent 2) For an even $n/e$;
\begin{itemize}
	\item[] $|\mathcal{I}_{00}|=(p^n+2p^e-3)/(2(p^e-1))$
	\item[] $|\mathcal{U}_{00}(b)|=\begin{cases}\frac{p^e-3}{4},&~{\rm for~}b=\pm 1{~\rm and~}p^e\equiv 3\bmod 4\\ \frac{p^e-5}{4},&~{\rm for~}b=1{~\rm and~}p^e\equiv 1\bmod 4 \\ \frac{p^e-1}{4},&~{\rm for~}b=-1{~\rm and~}p^e\equiv 1\bmod 4 \\ \frac{p^e-1}{2},&~{\rm for~}b(\neq \pm1)\in \mathcal{I}_{00}.\end{cases}$
\end{itemize}
\end{lemma}
\vspace{0.1cm}
\begin{proof}
From Lemma \ref{Lem_Cyc}, $D_f(x)|_{E_{00}}$ is represented in terms of $t$ as
\begin{align}\label{Eq_Dfx_C00}
	D_f(x)|_{E_{00}}=\frac{\alpha^{(p^k-1)t}+\alpha^{-(p^k-1)t}}{2}\triangleq {M}\big(\alpha^{(p^k-1)t}\big)
\end{align}
where $x=(\alpha^t-\alpha^{-t})^2/4$ and $t$ varies over $\mathcal{T}_1$ for $p^n\equiv 3\bmod 4$ and $\mathcal{T}_2$ for $p^n\equiv 1\bmod 4$. Assume that there exist $x_1$ and $x_2(\neq x_1)$ in $E_{00}$ such that $D_f(x_1)=D_f(x_2)$. Let $t_1=\theta^{-1}(x_1)$ and $t_2=\theta^{-1}(x_2)$. From (\ref{Eq_Dfx_C00}), it is straightforward that $t_1$ and $t_2$ satisfy either
\begin{align}\label{Condition1-1} 
	t_1+ t_2\equiv 0 \bmod \frac{p^n-1}{p^e-1}
\end{align}
or
\begin{align}\label{Condition1-2} 
	t_1\equiv t_2 \bmod \frac{p^n-1}{p^e-1}.
\end{align}
Define the set 
\begin{align}\label{Def_S_mu}
S_{\mu}=
\begin{cases}
\{t\equiv \pm \mu \bmod v ~|~ t\in \mathcal{T}_1\},&~{\rm for~}p^n\equiv 3\bmod 4 \\
\{t\equiv \pm \mu \bmod v ~|~ t\in \mathcal{T}_2\},&~{\rm for~}p^n\equiv 1\bmod 4
\end{cases}
\end{align}
where $v=(p^n-1)/(p^e-1)$ and $0\leq \mu\leq \lfloor v/2 \rfloor$. Then, from (\ref{Condition1-1}) and (\ref{Condition1-2}), all the elements in $S_{\mu}$ give a single value ${ M}(\alpha^{(p^k-1)t})$ in $\mathcal{I}_{00}$ and the elements in each $S_{\mu}$ give   distinct values in $\mathcal{I}_{00}$. 

Therefore, $|\mathcal{I}_{00}|$ is equal to the number of distinct sets $S_{\mu}$'s. Since $0\leq \mu\leq \lfloor v/2 \rfloor$, $|\mathcal{I}_{00}|$ is equal to $(v+1)/2$ for odd $v$ and $v/2+1$ for even $v$. Note that $v$ is even when $n/e$ is even and odd when $n/e$ is odd.

Clearly, $S_{\mu}$ corresponds to $\mathcal{U}_{00}({ M}(\alpha^{(p^k-1)t}))$. Thus, obtaining $|\mathcal{U}_{00}(b)|$ for $b\in \mathcal{I}_{00}$ is finding out the cardinality of corresponding $S_{\mu}$, which can be done easily by applying Lemma \ref{Lem_Remainder}.

Now, in the case when $p^n\equiv 3\bmod 4$, we have
\begin{align*}
	S_{\mu}=\{t\equiv \pm \mu \bmod v~ |~ t\in \mathcal{T}_1\}.
\end{align*}
Since $\frac{p^n-3}{4}=\frac{p^e-3}{4}v+\frac{v-1}{2}$, from Lemma \ref{Lem_Remainder}, we have
\begin{align*}
	|S_{\mu}|=
	\begin{cases}
		\frac{p^e-1}{2},&~{\rm for~}0<\mu < \frac{v}{2}\\
		\frac{p^e-3}{4},&~{\rm for~}\mu=0.
	\end{cases}
\end{align*}
Since $n/e$ is odd, i.e., $v$ is odd, in this case, we don't need to consider $S_{v/2}$. Note that $S_0$ corresponds to $\mathcal{U}_{00}(1)$.

Similarly, in the case when $p^n\equiv 1\bmod 4$, we have
\begin{align*}
	S_{\mu}=\{t\equiv \pm \mu \bmod v~|~t\in \mathcal{T}_2\}.
\end{align*}
Clearly, $p^e$ can be congruent to $3$ or $1$ modulo $4$ in this case. Since $\frac{p^n-5}{4}=\frac{p^e-3}{4}v+(\frac{v}{2}-1)$ for $p^e \equiv 3\bmod 4$, from Lemma \ref{Lem_Remainder}, we have
\begin{align*}
	|S_{\mu}|=
	\begin{cases}
		\frac{p^e-1}{2},&~{\rm for~}0<\mu < \frac{v}{2}\\
		\frac{p^e-3}{4},&~{\rm for~}\mu=0~{\rm~or~}\frac{v}{2}.
	\end{cases}
\end{align*}
Note that $n/e$ is even, i.e., $v$ is even, in this case, $S_{v/2}$ should be considered.  Note that $S_0$ corresponds to $\mathcal{U}_{00}(1)$ and $S_{{v}/{2}}$ corresponds to $\mathcal{U}_{00}(-1)$. Since $\frac{p^n-5}{4}=\frac{p^e-5}{4}v+(v-1)$ for $p^e \equiv 1\bmod 4$, from Lemma \ref{Lem_Remainder}, we have
\begin{align*}
	|S_{\mu}|=
	\begin{cases}
		\frac{p^e-1}{2},&~{\rm for~}0<\mu < \frac{v}{2}\\
		\frac{p^e-5}{4},&~{\rm for~}\mu=0\\
		\frac{p^e-1}{4},&~{\rm for~}\mu=\frac{v}{2} {\rm ~and~ even~} \frac{n}{e}.
	\end{cases}
\end{align*}
Here, $S_{v/2}$ should be considered only when $n/e$ is even. Note that $S_0$ corresponds to $\mathcal{U}_{00}(1)$ and $S_{{v}/{2}}$ corresponds to $\mathcal{U}_{00}(-1)$.  \hfill $\Box$
\end{proof}

\begin{lemma}
\label{Lem_MapProperty_11}
For $\mathcal{I}_{11}$ and $D_f|_{E_{11}}$, we have

\noindent 1) For an odd $n/e$;
\begin{itemize}
	\item[] $|\mathcal{I}_{11}|=(p^n+p^e-2)/(2(p^e-1))$
	\item[] $|\mathcal{U}_{11}(b)|=\begin{cases}\frac{p^e-3}{4},&~{\rm for~}b=1,~p^n\equiv 3\bmod 4,~{\rm even~}k/e\\&~~~~~{\rm or~}b=-1,~p^n\equiv 3\bmod 4,~{\rm odd~}k/e\\ \frac{p^e-1}{4},&~{\rm for~}b=1,~p^n\equiv 1\bmod 4,~{\rm even~}k/e\\&~~~~~{\rm or~}b=-1,~p^n\equiv 1\bmod 4,~{\rm odd~}k/e \\ \frac{p^e-1}{2},&~{\rm for~remaining~}b\in \mathcal{I}_{11}.\end{cases}$
\end{itemize}

\noindent 2) For an even $n/e$;
\begin{itemize}
	\item[] $|\mathcal{I}_{11}|=(p^n-1)/(2(p^e-1))$
	\item[] $|\mathcal{U}_{11}(b)|=(p^e-1)/2$ {\rm for~any} $b\in \mathcal{I}_{11}$.
\end{itemize}
\end{lemma}
\begin{proof}
~~\\
{\bf Case 1)} For $p^n\equiv 3\bmod 4$;
\vspace{0.1cm}

By selecting $\gamma=-1$ in (\ref{Eq_C11}), $D_f(x)|_{E_{11}}$ is represented as
\begin{align}\label{Eq_0423_001}
	D_f(x)|_{E_{11}}={M}\big((-1)^{\frac{p^k-1}{2}}\alpha^{(p^k-1)t}\big)=(-1)^{\frac{p^k-1}{2}}{ M}\big(\alpha^{(p^k-1)t}\big)
\end{align}
where $t\in \mathcal{T}_1$ and $x=-(\alpha^t+\alpha^{-t})^2/4$. 

Since 
\begin{align*}
	(-1)^{\frac{p^k-1}{2}}=
	\begin{cases}
		1,&~{\rm if~}p^k\equiv 1\bmod 4\\
		-1,&~{\rm if~}p^k\equiv 3\bmod 4
	\end{cases}
\end{align*}
and $t$ varies over $\mathcal{T}_1$, we have $\mathcal{I}_{00}=\mathcal{I}_{11}$ for $p^k\equiv 1\bmod 4$ and $\mathcal{I}_{00}=-\mathcal{I}_{11}$ for $p^k\equiv 3\bmod 4$. Therefore, $|\mathcal{I}_{11}|$ and $|\mathcal{U}_{11}(b)|$ are equal to $|\mathcal{I}_{00}|$ and $|\mathcal{U}_{00}(b)|$ in Lemma  \ref{Lem_MapProperty_00}, respectively. Note that $n/e$ is odd in this case.

\vspace{0.1cm}
\noindent {\bf Case 2)} For $p^n\equiv 1\bmod 4$;
\vspace{0.1cm}

In this case, we select $\gamma=-\alpha$. Then $D_f(x)|_{E_{11}}$ is represented as 
\begin{align}\label{Eq_0423_002}
	D_f(x)|_{E_{11}}={M}\big((-\alpha)^{\frac{p^k-1}{2}}\alpha^{(p^k-1)t}\big)
\end{align}
where $t\in \mathcal{T}_2\cup \{0\}$ and $x=-\alpha(\alpha^t+\alpha^{-(t+1)})^2/4$. 

Assume that $D_f(x_1)=D_f(x_2)$ for two distinct elements $x_1$ and $x_2$ in $E_{11}$. Let $t_1=\theta^{-1}(x_1)$ and $t_2=\theta^{-1}(x_2)$. Then, from (\ref{Eq_0423_002}), $t_1$ and $t_2$ should satisfy
\begin{align}\label{Eq_0423_003}
	t_1+t_2+1\equiv 0\bmod v
\end{align}
or
\begin{align}\label{Eq_0423_0031}
	t_1\equiv t_2\bmod v
\end{align}
where $v=({p^n-1})/({p^e-1})$. 

Note that $\mathcal{T}_2\cup \{0\}\cong \mathbb{Z}_{\frac{p^n-1}{4}}$. Let $R_i$, $0\leq i\leq v-1$, be the equivalent class congruent to $i$ modulo $v$ in $\mathbb{Z}_{\frac{p^n-1}{4}}$. 

From (\ref{Eq_0423_003}) and (\ref{Eq_0423_0031}), we know that all the elements $t$ in $R_i\cup R_{v-i-1}$ map to a single value in $\mathcal{I}_{11}$. Thus, obtaining $|\mathcal{U}_{11}(b)|$ is just finding out the corresponding $|R_i\cup R_{v-i-1}|$. When $v$ is odd, i.e., $n/e$ is odd, and $i=(v-1)/2$, $R_i$ coincides with $R_{v-i-1}$. In this case, we can easily check that any $t$ in $R_{({v-1})/{2}}$ maps to $1$ for even $k/e$ and $-1$ for odd $k/e$. Otherwise, $|\mathcal{U}_{11}(b)|=(p^e-1)/2$, since $|R_i|=(p^e-1)/4$. 
\hfill $\Box$
\end{proof}

\begin{lemma}
\label{E01_E10}
For $\mathcal{I}_{10}$, $\mathcal{I}_{01}$, $D_f|_{E_{10}}$, and $D_f|_{E_{01}}$, we have

\noindent 1) For an odd $k/e$;
\begin{itemize}
	\item[] $D_f$ is bijective on both $E_{10}$ and $E_{01}$ so that $|\mathcal{I}_{10}|=|E_{10}|=(1,0)$ and $|\mathcal{I}_{01}|=|E_{01}|=(0,1)$.
	\item[] $1\not\in \mathcal{I}_{10}$ and $1\not\in \mathcal{I}_{01}$. 
\end{itemize}

\noindent 2) For an even $k/e$;
\begin{itemize}
	\item[] $|\mathcal{I}_{10}|=|\mathcal{I}_{01}|=(p^n+p^e+2)/(2(p^e+1))$
	\item[] $|\mathcal{U}_{10}(b)|=\begin{cases}\frac{p^e-1}{4},&~{\rm for~}b=1,~p^n\equiv 1\bmod 4 \\ \frac{p^e-3}{4},&~{\rm for~}b=1,~p^n\equiv 3\bmod 4 \\ \frac{p^e+1}{2},&~{\rm for~}b(\neq 1)\in \mathcal{I}_{10}\end{cases}$
	\item[] $|\mathcal{U}_{01}(b)|=\begin{cases}\frac{p^e-1}{4},&~{\rm for~}b=1,~p^n\equiv 1\bmod 4 \\ \frac{p^e+1}{4},&~{\rm for~}b=1,~p^n\equiv 3\bmod 4 \\ \frac{p^e+1}{2},&~{\rm for~}b(\neq 1)\in \mathcal{I}_{01}.\end{cases}$
\end{itemize}
\end{lemma}
\begin{proof} 
~~~\\
\noindent {\bf Case 1)} For $I_{10}$ and $D_f|_{E_{10}}$;
\vspace{0.1cm}

From Lemma \ref{Lem_Cyc}, $D_f(x)|_{E_{10}}$ can be written as
\begin{align}\label{0520_001}
	D_f(x)|_{E_{10}}=M\big(\beta^{(p^k-1)(p^n-1)t}\big)=M\big(\delta^{2(p^k-1)t}\big)
\end{align}
where $x=(\delta^{2t}-\delta^{-2t})^2/4$ and $t$ varies over $[1,(p^n-3)/4]$ for $p^n\equiv 3\bmod 4$ and over $[1,(p^n-1)/4]$ for $p^n\equiv 1\bmod 4$. 

Let $t=\theta^{-1}(x)$. Then from (\ref{0520_001}), $\theta(t_1)$ and $\theta(t_2)$ give the same value of $D_f(x)$ if and only if
\begin{align}\label{0520_002}
	t_1\pm t_2\equiv 0 \bmod L
\end{align}
where $L=(p^n+1)/\gcd(p^k-1,p^n+1)$. From Lemma \ref{Lem_gcd}, $L=(p^n+1)/2$ for odd $k/e$ and $L=(p^n+1)/(p^e+1)$ for even $k/e$. 

Now, consider the case when $p^n\equiv 3\bmod 4$ and odd $k/e$. Since $\mathcal{T}_1=[1,(p^n-3)/4]$, no $t_1$ and $t_2$ in $\mathcal{T}_1$ satisfy (\ref{0520_002}) so that $D_f$ is bijective on $E_{10}$. Note that there exists no $t\in \mathcal{T}_1$ such that $t\bmod L\equiv 0$, that is, $D_f(x)\neq 1$. Hence we can conclude that $|\mathcal{I}_{10}|=|E_{10}|=(p^n-3)/4$ and $1\not\in\mathcal{I}_{10}$. 

For the case when $p^n\equiv 3\bmod 4$ and even $k/e$, we can use Lemma \ref{Lem_Remainder} by setting $v=L=(p^n+1)/(p^e+1)$. In this case, $q$ and $r$ become $q=(p^e-3)/4$ and $r=v-1$. Note that $v$ is odd in this case. From Lemma \ref{Lem_Remainder}, it is derived that $|\mathcal{U}_{10}(b)|=(p^e+1)/2$ for $b(\neq 1)\in \mathcal{I}_{10}$  and $|\mathcal{U}_{10}(1)|=(p^e-3)/4$. For the case when $p^n\equiv 1\bmod 4$, the proof can be done similarly.

\vspace{0.1cm}
\noindent {\bf Case 2)} For $I_{01}$ and $D_f|_{E_{01}}$;
\vspace{0.1cm}
  
  In this case, $D_f(x)|_{E_{01}}$ can be written as
  \begin{align*}
  	D_f(x)|_{E_{01}}=M\big(\delta^{(2t+1)(p^k-1)}\big)
  \end{align*}
  where $x=(\delta^{2t+1}-\delta^{-(2t+1)})^2/4$ and $t$ varies over $[0,(p^n-3)/4]$ for $p^n\equiv 3\bmod 4$ and over $[0,(p^n-5)/4]$ for $p^n\equiv 1\bmod 4$. 
  
  Using the similar argument to the previous case, $\theta(t_1)$ and $\theta(t_2)$ give the same value of $D_f(x)$ if and only if
\begin{align}\label{0520_003}
	(2t_1+1)(p^k-1)\pm (2t_2+1)(p^k-1) \equiv 0 \bmod 2(p^n+1). 
\end{align}  
Then (\ref{0520_003}) can be rewritten as either
\begin{align}\label{0520_004}
	t_1-t_2\equiv 0 \bmod L
\end{align}
or
\begin{align}\label{0520_005}
	t_1+t_2+1\equiv 0 \bmod L.
\end{align}
For the case when $p^n\equiv 3\bmod 4$ and odd $k/e$, again $D_f$ is bijective on $E_{01}$ and there exists no $x$ such that $D_f(x)=1$. Thus, $|\mathcal{I}_{01}|=|E_{01}|=(p^n+1)/4$ and $1\not\in \mathcal{I}_{01}$. 

For the case when $p^n\equiv 3\bmod 4$ and even $k/e$, applying Lemma \ref{Lem_Remainder} to (\ref{0520_004}) and (\ref{0520_005}) yields that $|\mathcal{U}_{01}(b)|=(p^e+1)/2$ for $b(\neq 1)\in \mathcal{I}_{01}$  and $|\mathcal{U}_{01}(1)|=(p^e+1)/4$. For the case when $p^n\equiv 1\bmod 4$, the proof can be done similarly. 

\hfill $\Box$
\end{proof}

So far, we have investigated the cardinality of the images and the inverse images of $D_f|_{E_{ij}}$, $i,j\in \{0,1\}$. In order to unify Lemmas \ref{Lem_MapProperty_00}--\ref{E01_E10} and see the overall mapping property of $D_f$, we have to look into the relationship between  $\mathcal{I}_{00}$, $\mathcal{I}_{11}$, $\mathcal{I}_{10}$, and $\mathcal{I}_{01}$ as in the following three lemmas. 

\begin{lemma}
\label{Relation_00_11}
For $\mathcal{I}_{00}$ and $\mathcal{I}_{11}$, we have
  
\begin{align*}
	\begin{cases}
		\mathcal{I}_{00}=\mathcal{I}_{11},&~{\rm for~even~}\frac{k}{e}\\
		\mathcal{I}_{00}\cap \mathcal{I}_{11}=\emptyset,&~{\rm for~odd~}\frac{k}{e}.
	\end{cases}
\end{align*}
\end{lemma}
\begin{proof}
~~~\\ 
 \noindent {\bf Case 1)} For $p^n \equiv 3\bmod 4$;
 \vspace{0.1cm}
 
 In this case, $k/e$ is even when $p^k\equiv 1\bmod 4$ and $k/e$ is odd when $p^k\equiv 3\bmod 4$. In Lemma \ref{Lem_MapProperty_11}, we already showed that $\mathcal{I}_{00}=\mathcal{I}_{11}$ for $p^k\equiv 1\bmod 4$ and $\mathcal{I}_{00}=-\mathcal{I}_{11}$ for $p^k\equiv 3\bmod 4$. Thus, the remaining part is to show that any two elements $a$ and $-a$ cannot belong to $\mathcal{I}_{00}$. Assume that there are two distinct elements $x_1$ and $x_2$ in $E_{00}$ such that $D_f(x_1)=-D_f(x_2)$. Let $t_1=\theta^{-1}(x_1)$ and $t_2=\theta^{-1}(x_2)$. Then from (\ref{Eq_Dfx_C00}), it is easy to see that either $\alpha^{(p^k-1)t_1}=-\alpha^{(p^k-1)t_2}$ or $\alpha^{(p^k-1)t_1}=-\alpha^{-(p^k-1)t_2}$ must hold. But this is a contradiction because $-\alpha^{\pm (p^k-1)t_2}$ is a nonsquare in $\mathbb{F}_{p^n}$, whereas $\alpha^{(p^k-1)t_1}$ is a square in $\mathbb{F}_{p^n}$. Therefore, $\mathcal{I}_{00}\cap \mathcal{I}_{11}=\mathcal{I}_{00}
 \cap (-\mathcal{I}_{00})=\emptyset$ for odd $k/e$.  
 
 \vspace{0.1cm}
 \noindent {\bf Case 2)} For $p^n \equiv 1\bmod 4$; 
 \vspace{0.1cm}

 Again, assume that there exist $x_1\in E_{00}$ and $x_2\in E_{11}$ such that $D_f(x_1)=D_f(x_2)$. Let $t_1=\theta^{-1}(x_1)\in \mathcal{T}_2$ and $t_2=\theta^{-1}(x_2)\in \mathcal{T}_2\cup \{0\}$. Then, from (\ref{Eq_Dfx_C00}) and (\ref{Eq_0423_002}), we have
 \begin{align}\label{0423_1}
 	(-\alpha)^{\frac{p^k-1}{2}}\alpha^{(p^k-1)t_2}=\alpha^{(p^k-1)t_1}{~\rm or~}\alpha^{-(p^k-1)t_1}.
 \end{align}
 For $p^k\equiv 3\bmod 4$, (\ref{0423_1}) cannot be satisfied because the left-hand side of (\ref{0423_1}) is a nonsquare in $\mathbb{F}_{p^n}$, while the right-hand side of (\ref{0423_1}) is a square in $\mathbb{F}_{p^n}$. 
 
 For $p^k\equiv 1\bmod 4$, (\ref{0423_1}) implies that either $\alpha^{\frac{p^k-1}{2}(2t_1+2t_2+1)}=1$ or $\alpha^{\frac{p^k-1}{2}(2t_2-2t_1+1)}=1$, which further implies that either $2(t_1+t_2)+1$ or $2(t_2-t_1)+1$ must be divisible by $2(p^n-1)/(p^e-1)$ for odd $k/e$ and $(p^n-1)/(p^e-1)$ for even $k/e$.

 Since $2(t_2\pm t_1)+1$ is odd, we can easily see that the above is possible only when $k/e$ is even and $n/e$ is odd and that such $t_1$ and $t_2$ can be always found in $\mathcal{T}_2$ and $\mathcal{T}_2\cup \{0\}$, respectively. Note that $n/e$ is always odd when $k/e$ is even. Hence we conclude that $\mathcal{I}_{00}=\mathcal{I}_{11}$ for even $k/e$ and $\mathcal{I}_{00}\cap \mathcal{I}_{11}=\emptyset$, otherwise.   \hfill $\Box$
\end{proof}

\begin{lemma}
\label{Relation_01_10}
For $\mathcal{I}_{10}$ and $\mathcal{I}_{01}$, we have
\begin{align*}
	\begin{cases}
		\mathcal{I}_{10}\cap \mathcal{I}_{01}=\emptyset,&~{\rm for~odd~}\frac{k}{e}~{\rm and~}p^e\equiv 3\bmod 4\\
		\mathcal{I}_{10}= \mathcal{I}_{01},&~{\rm otherwise}.
	\end{cases}
\end{align*}
\end{lemma}
\begin{proof}
Assume that there exist $x_1\in E_{10}$ and $x_2\in E_{01}$ such that $D_f(x_1)=D_f(x_2)$. Let 
$t_1=\theta^{-1}(x_1)$ and $t_2=\theta^{-1}(x_2)$. Then, from Lemma \ref{Lem_Cyc}, we have
\begin{align}\label{Eq_0413_1}
	\delta^{2t_1(p^k-1)}+\delta^{-2t_1(p^k-1)}=\delta^{(2t_2+1)(p^k-1)}+\delta^{-(2t_2+1)(p^k-1)}.
\end{align}
Since $\delta^{2(p^n+1)}=1$, the necessary and sufficient conditions for (\ref{Eq_0413_1}) to hold is
\begin{align}\label{0521_001}
	2(t_2\pm t_1)+1\equiv 0 \bmod L
\end{align} 
where $L=2(p^n+1)/\gcd(2(p^n+1),p^k-1)$.

Note that $t_1$ lies in $[1,(p^n-3)/4]$ for $p^n\equiv 3\bmod 4$ and in $[1,(p^n-1)/4]$ for $p^n\equiv 1\bmod 4$ and $t_2$ lies in $[0,(p^n-3)/4]$ for $p^n\equiv 3\bmod 4$ and in $[0,(p^n-5)/4]$ for $p^n\equiv 1\bmod 4$. 

When $L$ becomes even, which occurs only if $k/e$ is odd and $p^e\equiv 3\bmod 4$, (\ref{0521_001}) cannot be satisfied because the left-hand side of (\ref{0521_001}) is odd. Hence we conclude that $\mathcal{I}_{01}\cap \mathcal{I}_{10}=\emptyset$ in this case. 

Otherwise, it is not difficult to find $t_2$ satisfying (\ref{0521_001}) for each $t_1$ because  $L$ is either $(p^n+1)/2$ or $(p^n+1)/(p^e+1)$ which is odd. Since $|\mathcal{I}_{10}|=|\mathcal{I}_{01}|$ in Lemma \ref{E01_E10}, the proof is done.  \hfill $\Box$
\end{proof}

\begin{lemma}
\label{Lemm_Relationship}

Let $\mathcal{S}_1=\mathcal{I}_{00} \cup \mathcal{I}_{11}$ and $\mathcal{S}_2=\mathcal{I}_{01} \cup \mathcal{I}_{10}$. Then we have
\begin{align*}
	\mathcal{S}_1\cap \mathcal{S}_2=
	\begin{cases}
		\emptyset,&~{\rm for~odd~}\frac{k}{e}\\
		\{1\},&~{\rm for~even~}\frac{k}{e}.
	\end{cases}
\end{align*}
\begin{proof}
The proof is in Appendix. 
\end{proof}

\end{lemma}

Using the previous lemmas, the main theorem can be stated as follows. 

\begin{theorem}
{\label{Thm_1}
For an odd prime $p$ and $d=(p^k+1)/2$, the differential spectrum of the function $f(x)=x^d$ in $\mathbb{F}_{p^n}$ is given as:

\vspace{0.1cm}
\noindent 1) For an odd $k/e$;

\vspace{0.1cm}
1-i) For $p^e\equiv 3\bmod 4$;
\begin{align*}
	\omega_i=
	\begin{cases}
		2,& ~{\rm if~}i=\frac{p^e+1}{4}~({\rm the~corresponding~two~}b'{\rm s~are~}\pm 1)\\
		\frac{p^n-p^e}{p^e-1},&~{\rm if~}i=\frac{p^e-1}{2} \\
		\frac{p^n-1}{2},&~{\rm if~}i=1\\
		\frac{p^n-3}{2}-\frac{p^n-p^e}{p^e-1},&~{\rm if~}i=0\\
		0,&~{\rm otherwise}.
	\end{cases}
\end{align*}

1-ii) For $p^e\equiv 1\bmod 4$;
\begin{align*}
	\omega_i=
	\begin{cases}
		1,& ~{\rm if~}i=\frac{p^e+3}{4}~({\rm the~corresponding~}b{\rm ~is~} 1)\\
		1,& ~{\rm if~}i=\frac{p^e-1}{4}~({\rm the~corresponding~}b{\rm ~is~} -1)\\
		\frac{p^n-p^e}{p^e-1},& ~{\rm if~}i=\frac{p^e-1}{2}\\
		\frac{p^n-1}{4},& ~{\rm if~}i=2\\
		\frac{(p^n-1)(3p^e-7)}{4(p^e-1)},& ~{\rm if~}i=0\\
		0,&~{\rm otherwise}.
	\end{cases}
\end{align*}

\vspace{0.1cm}
2) For an even $k/e$;
\begin{align*}
	\omega_i=
	\begin{cases}
		1,& ~{\rm if~}i=p^e~({\rm the~corresponding~}b{\rm ~is~} 1)\\
		\frac{p^n-p^e}{2(p^e-1)},& ~{\rm if~}i={p^e-1}\\
		\frac{p^n-p^e}{2(p^e+1)},& ~{\rm if~}i={p^e+1}\\
		p^n-\frac{p^{n+e}-1}{p^{2e}-1},& ~{\rm if~}i=0\\
		0,&~{\rm otherwise}
	\end{cases}
\end{align*}
where $e=\gcd(n,k)$.
}
\end{theorem}
\begin{proof} So far, we have derived $|\mathcal{I}_{ij}|$'s and $|\mathcal{U}_{ij}(b)|$'s in Lemmas \ref{Lem_MapProperty_00}--\ref{E01_E10}. From Lemmas \ref{Relation_00_11} and \ref{Relation_01_10}, we have seen that $\mathcal{I}_{00}$ and $\mathcal{I}_{11}$ are either disjoint or identical and so be $\mathcal{I}_{01}$ and $\mathcal{I}_{10}$. Finally, from Lemma \ref{Lemm_Relationship}, we have seen that $\mathcal{I}_{00}\cup \mathcal{I}_{11}$ and $\mathcal{I}_{01}\cup \mathcal{I}_{10}$ are either disjoint or almost disjoint. For the proof of this theorem, we have to combine these results.

\vspace{0.1cm}
	\noindent {\bf Case 1)} Combining $D_f|_{E_{00}}$ and $D_f|_{E_{11}}$;
	\vspace{0.1cm}	
	
	For the case when $k/e$ is odd, we have $|\mathcal{I}_{00}\cup \mathcal{I}_{11}|=(p^n+p^e-2)/(p^e-1)$ because $\mathcal{I}_{00}$ and $\mathcal{I}_{11}$ are disjoint. For any $b\in(\mathcal{I}_{00}\cup \mathcal{I}_{11})\setminus\{1,-1\}$, the cardinality of $\mathcal{U}_0(b)$, the inverse image in $E_{00}\cup E_{11}$ of $b$, is $(p^e-1)/2$. For the elements $\pm 1\in \mathcal{I}_{00}\cup \mathcal{I}_{11}$, we have
	\begin{align*}
		(|\mathcal{U}_0(1)|,|\mathcal{U}_0(-1)|)=
		\begin{cases}
			\big(\frac{p^e-5}{4},\frac{p^e-1}{4}\big),~{\rm for~}p^e\equiv 1\bmod 4\\
			\big(\frac{p^e-3}{4},\frac{p^e-3}{4}\big),~{\rm for~}p^e\equiv 3\bmod 4.
		\end{cases}
	\end{align*}
	
	For the case when $k/e$ is even, we have $|\mathcal{I}_{00}\cup \mathcal{I}_{11}|=(p^n+p^e-2)/(2(p^e-1))$ since $\mathcal{I}_{00}$ and $\mathcal{I}_{11}$ coincide. Also, we have
\begin{align*}
	|\mathcal{U}_0(b)|=
	\begin{cases}
		p^e-1,&~{\rm if~}b\in (\mathcal{I}_{00}\cup \mathcal{I}_{11})\setminus\{1\}\\
		\frac{p^e-3}{2},&~{\rm if~}b=1.
	\end{cases}
\end{align*}
	
	\noindent {\bf Case 2)} Combining $D_f|_{E_{10}}$ and $D_f|_{E_{01}}$;
	\vspace{0.1cm}

	For the case when $k/e$ is odd, we have $|\mathcal{I}_{10}\cup \mathcal{I}_{01}|=|\mathcal{I}_{10}|=|\mathcal{I}_{01}|=(p^n-1)/4$ for $p^e\equiv 1\bmod 4$ and $|\mathcal{I}_{10}\cup \mathcal{I}_{01}|=|\mathcal{I}_{10}|+|\mathcal{I}_{01}|=(p^n-1)/2$ for $p^e\equiv 3\bmod 4$. The cardinality of $\mathcal{U}_1(b)$, the inverse image in $E_{10}\cup E_{01}$ of $b\in (\mathcal{I}_{10}\cup \mathcal{I}_{01})$, is 
	\begin{align*}
		|\mathcal{U}_1(b)|=
		\begin{cases}
			2,&~{\rm for~}p^e\equiv 1\bmod 4\\
			1,&~{\rm for~}p^e\equiv 3\bmod 4.
		\end{cases}
	\end{align*}

	For the case when $k/e$ is even, we have $|\mathcal{I}_{10}\cup \mathcal{I}_{01}|=|\mathcal{I}_{10}|=|\mathcal{I}_{01}|=(p^n+p^e+2)/(2(p^e+1))$. Also, we have
	\begin{align*}
		|\mathcal{U}_1(b)|=
		\begin{cases}
			p^e+1,&~{\rm if~}b\in (\mathcal{I}_{10}\cup \mathcal{I}_{01})\setminus \{1\}\\\
			\frac{p^e-1}{2},&~{\rm if~}b=1.
		\end{cases}
	\end{align*}

	The unified mapping property of $D_f|_{E_{10}}$ and $D_f|_{E_{01}}$ is that the cardinality of the inverse image  in $E_{10}\cup E_{01}$ of each element in $(\mathcal{I}_{10}\cup \mathcal{I}_{01})\setminus \{1\}$ is $p^e+1$ and the cardinality of the inverse image  in $E_{10}\cup E_{01}$ of the element $1\in \mathcal{I}_{10}\cup \mathcal{I}_{01}$ is $(p^e-1)/2$. 
	
	Since $x=0,-1\not\in (E_{00}\cup E_{11}\cup E_{10}\cup E_{01})$, we have to consider the case when $x=0$ and $x=-1$. It is easy to derive that $D_f(0)=1$ and $D_f(-1)=(-1)^{(p^k+3)/2}$. Finally, with Lemma \ref{Lemm_Relationship}, we can combine the Case 1) and Case 2). Hence the proof is done. \hfill $\Box$
\end{proof}

\begin{corollary}
\label{Mod_Helleseth}
For an odd prime $p$ and $d=(p^k+1)/2$, the differential uniformity $\Delta_f$ of $f(x)=x^d$ in $\mathbb{F}_{p^n}$ is given as
\begin{align*}
	\Delta_f=
	\begin{cases}
		\frac{p^e-1}{2},&~{\rm for~odd~}\frac{k}{e}\\
		p^e+1,&~{\rm for~even~}\frac{k}{e}
	\end{cases}
\end{align*}
where $e=\gcd(n,k)$. \hfill $\Box$  
\end{corollary}

The comparison with the existing bound in Theorem \ref{Thm_Helleseth} and our new result in Corollary \ref{Mod_Helleseth} is given in Table \ref{table1}. The bound in Theorem \ref{Thm_Helleseth} is not tight for some cases of $d=(p^k+1)/2$, whereas Theorem \ref{Thm_1}  provides the exact differential spectrum and $\Delta_f$ for $d=(p^k+1)/2$. We can also explain some known PN and APN functions which belong to this function class.

\begin{table*}[t]
\begin{center}
{\caption{Comparison between the existing bound in Theorem \ref{Thm_Helleseth} and new result in Corollary \ref{Mod_Helleseth}}\label{table1}
\centering
\begin{tabular}{|c|c|c|c|c|}
\hline
  $p$	&  $n$	&  $k$	& Upper bound on $\Delta_f$ in \cite{Helleseth}	&   Explicit $\Delta_f$ (new result) \\
\hline
\hline
  $5$	&  $3$	&  $2$	& $12$	&   $6$\\
\hline
  $5$	&  $5$	&  $2$	& $12$	&   $6$\\
\hline
  $5$	&  $5$	&  $4$	& $24$	&   $6$\\
\hline
  $7$	&  $3$	&  $2$	& $24$	&   $8$\\
\hline
 $7$	&  $5$	&  $2$	& $24$	&   $8$\\
\hline
 $7$	&  $5$	&  $4$	& $48$	&   $8$\\
\hline
 $7$	&  $7$	&  $2$	& $24$	&   $8$\\
\hline
 $7$	&  $7$	&  $4$	& $48$	&   $8$\\
\hline
 $7$	&  $7$	&  $6$	& $24$	&   $8$\\
\hline
 $11$	&  $3$	&  $2$	& $60$	&   $12$\\
\hline
\end{tabular}
}
\end{center}
\end{table*}

\section{The Differential Spectrum of $x^{\frac{p^n+1}{p^k+1}+\frac{p^n-1}{2}}$ in $\mathbb{F}_{p^n}$}
\label{Dif_Func2}

In this section, we consider the power function $f(x)=x^d$ with the power
\begin{align*}
	d=\frac{p^n+1}{p^k+1}+\frac{p^n-1}{2}
\end{align*}
where $n/k$ should be odd. Note that only when $p^k\equiv 3\bmod 4$, i.e., $p\equiv 3\bmod 4$ and $n$ is odd, there exists no inverse $d^{-1}=(p^k+1)/2$ which belongs to the function in the previous section. Hence,  for an odd prime $p$ such that $p\equiv 3\bmod 4$ and odd $n$ with $k|n$, we calculate the differential uniformity and the differential spectrum of the power function $f(x)=x^{\frac{p^n+1}{p^k+1}+\frac{p^n-1}{2}}$ in $\mathbb{F}_{p^n}$. 

Define the functions $h_i(x)$ in $\mathbb{F}_{p^n}$, $1\leq i\leq 4$, as
\begin{align*}
	h_1(x)=&(x+1)^{\frac{p^k+1}{2}}+x^{\frac{p^k+1}{2}}\\
	h_2(x)=&(x+1)^{\frac{p^k+1}{2}}-x^{\frac{p^k+1}{2}}\\
	h_3(x)=&-(x+1)^{\frac{p^k+1}{2}}+x^{\frac{p^k+1}{2}}\\
	h_4(x)=&-(x+1)^{\frac{p^k+1}{2}}-x^{\frac{p^k+1}{2}}.
\end{align*}
Let $\lambda_{i}(b)$ and $\chi_{i}(b)$ be the number of solutions of  
\begin{align}\label{Eq_Thm2_add1} 
h_i(x)=b^{-\frac{p^k+1}{2}}
\end{align}
in $E_{00}$ and $E_{11}$, respectively. 

\begin{lemma}\label{Lem_Map2}
For $f(x)=x^{\frac{p^n+1}{p^k+1}+\frac{p^n-1}{2}}$ and $b\in \mathbb{F}_{p^n}^*$, $N(1,b)$ is determined as:

\vspace{0.1cm}
\noindent 1) For $b\neq \pm 1$;
\begin{align*}
	N(1,b)=
	\begin{cases}
		\lambda_{1}(b)+\lambda_{2}(b)+\lambda_{3}(b)+\lambda_{4}(b),&~{\rm for~}b\in C_0\setminus\{1\}\\
		\chi_1(b)+\chi_2(b)+\chi_3(b)+\chi_4(b),&~{\rm for~}b\in C_1\setminus\{-1\}.
	\end{cases}
\end{align*}

\noindent 2) For $b=\pm 1$;
\begin{align*}
	N(1,b)=
	\begin{cases}
		\lambda_{1}(b)+\lambda_{2}(b)+\lambda_{3}(b)+\lambda_{4}(b)+1,&~{\rm for~}b=1\\
		\chi_1(b)+\chi_2(b)+\chi_3(b)+\chi_4(b)+1,&~{\rm for~}b=-1.
	\end{cases}
\end{align*}
\end{lemma}
\begin{proof}
Consider the cases when $x\in \mathbb{F}_{p^n}^*\setminus\{-1\}$. Since $\gcd(p^k+1,p^n-1)=2$, an element $x\in E_{00}$ can be expressed as $x=\nu^{p^k+1}$ and $x+1=\psi^{p^k+1}$ for some $\nu$ and $\psi$. If this $x$ is a solution to $D_f(x)=b$, then we have
\begin{align}\label{Eq_0515_00}
	(x+1)^{\frac{p^n+1}{p^k+1}+\frac{p^n-1}{2}}-x^{\frac{p^n+1}{p^k+1}+\frac{p^n-1}{2}}=\psi^2-\nu^2=b.
\end{align}

By setting  $y=b^{-1}\nu^2$, we have $y+1=b^{-1}\psi^2$ and thus $y$ becomes the solution to 
\begin{align}\label{Eq_0515_001}
	(y+1)^{\frac{p^k+1}{2}}-y^{\frac{p^k+1}{2}}=b^{-\frac{p^k+1}{2}}.
\end{align}
Since the transformation $x$ to $y$ is one-to-one, each solution $x\in E_{00}$ to $D_f(x)=b$  corresponds to either a solution $y\in E_{00}$ to (\ref{Eq_0515_001}) for $b\in C_0$ or a solution $y\in E_{11}$ to (\ref{Eq_0515_001}) for $b\in C_1$.

Similarly, if $x\in E_{11}$ is a solution to $D_f(x)=b$, then by letting $x+1=-\psi^{p^k+1}$ and $x=-\nu^{p^k+1}$, we have (\ref{Eq_0515_00}). Again by setting $y=b^{-1}\nu^2$, we have $y+1=b^{-1}\psi^2$. Thus $y$ is a solution to 
\begin{align}\label{Eq_0515_002}
	-(y+1)^{\frac{p^k+1}{2}}+y^{\frac{p^k+1}{2}}=b^{-\frac{p^k+1}{2}}.
\end{align}
Since the transformation $x$ to $y$ is one-to-one, each solution $x\in E_{11}$ to $D_f(x)=b$  corresponds to either a solution $y\in E_{00}$ to (\ref{Eq_0515_002}) for $b\in C_0$, or a solution $y\in E_{11}$ to (\ref{Eq_0515_002}) for $b\in C_1$.

Similarly, if $x\in E_{10}$ is a solution to $D_f(x)=b$, then by letting $x+1=\psi^{p^k+1}$ and $x=-\nu^{p^k+1}$, we have (\ref{Eq_0515_00}). Again by setting $y=b^{-1}\nu^2$, we have $y+1=b^{-1}\psi^2$. Thus $y$ is a solution to 
\begin{align}\label{Eq_0515_003}
	(y+1)^{\frac{p^k+1}{2}}+y^{\frac{p^k+1}{2}}=b^{-\frac{p^k+1}{2}}.
\end{align}
Since the transformation $x$ to $y$ is one-to-one, each solution $x\in E_{11}$ to $D_f(x)=b$  corresponds to either a solution $y\in E_{00}$ to (\ref{Eq_0515_003}) for $b\in C_0$ or a solution $y\in E_{11}$ to (\ref{Eq_0515_003}) for $b\in C_1$.

Similarly, if $x\in E_{01}$ is a solution to $D_f(x)=b$, then by letting $x+1=-\psi^{p^k+1}$ and $x=\nu^{p^k+1}$, we have (\ref{Eq_0515_00}). Again by setting $y=b^{-1}\nu^2$, we have $y+1=b^{-1}\psi^2$. Thus $y$ is a solution to 
\begin{align}\label{Eq_0515_004}
	-(y+1)^{\frac{p^k+1}{2}}-y^{\frac{p^k+1}{2}}=b^{-\frac{p^k+1}{2}}.
\end{align}
Since the transformation $x$ to $y$ is one-to-one, each solution $x\in E_{11}$ to $D_f(x)=b$  corresponds to either a solution $y\in E_{00}$ to (\ref{Eq_0515_004}) for $b\in C_0$ or a solution $y\in E_{11}$ to (\ref{Eq_0515_004}) for $b\in C_1$.

Since $D_f(0)=1$ and $D_f(-1)=-1$, we have completed the proof. \hfill $\Box$   
\end{proof}

Using the above lemma, the differential spectrum of $f(x)$ can be derived as follows. 

\begin{theorem}
\label{Main_Thrm2}
For an odd prime $p$ such that $p\equiv 3\bmod 4$, odd $n$ with $k|n$, and $d=(p^n+1)/(p^k+1)+(p^n-1)/2$, the differential spectrum of $f(x)=x^d$ in $\mathbb{F}_{p^n}$ is given as
\begin{align*}
	\omega_i=
	\begin{cases}
		2,& ~{\rm if~}i=\frac{p^k+1}{4}~({\rm the~corresponding~two~}b'{\rm s~are~}\pm 1)\\
		\frac{p^n-p^k}{p^k-1},& ~{\rm if~}i=\frac{p^k+1}{2}\\
		\frac{p^n-1}{2}-\frac{p^n-p^k}{p^k-1},& ~{\rm if~}i=1\\
		\frac{p^n-3}{2},& ~{\rm if~}i=0\\
		0,&~{\rm otherwise}.
	\end{cases}
\end{align*}
\end{theorem}
\begin{proof} From Lemma \ref{Lem_Map2}, in order to determine $N(1,b)$, we should calculate $\sum_{i=1}^4\lambda_i(b)$ and $\sum_{i=1}^4\chi_i(b)$ for $b\in C_0$ and $b\in C_1$, respectively. 
 From Lemma \ref{Lem_Cyc}, $h_1(x)$ and $h_2(x)$ on $E_{00}$ can be represented as
\begin{align}\label{Thm2_Eq1}
	h_1(x)|_{E_{00}}=\frac{\alpha^{t(p^k+1)}+\alpha^{-t(p^k+1)}}{2}\nonumber\\
	h_2(x)|_{E_{00}}=\frac{\alpha^{t(p^k-1)}+\alpha^{-t(p^k-1)}}{2}
\end{align}
where $x=({\alpha^t-\alpha^{-t}})^2/4$ and $t$ varies over $\mathcal{T}_1$. 
Similarly, $h_1(x)$ and $h_2(x)$ on $E_{11}$ can be represented as
\begin{align}\label{Thm2_Eq2}
	h_1(x)|_{E_{11}}=\frac{\alpha^{t(p^k+1)}+\alpha^{-t(p^k+1)}}{2}\nonumber\\
	h_2(x)|_{ E_{11}}=-\frac{\alpha^{t(p^k-1)}+\alpha^{-t(p^k-1)}}{2}
\end{align}
where $x=-(\alpha^t+\alpha^{-t})^2/4$ and $t$ varies over $\mathcal{T}_1$. Note that $h_1(x)=-h_4(x)$ and $h_2(x)=-h_3(x)$. 

Since $\gcd((p^k+1)/2,p^n-1)=2$, $b^{-\frac{p^k+1}{2}}$ in (\ref{Eq_Thm2_add1}) varies over $C_0$ twice, while $b$ varies over $\mathbb{F}_{p^n}^*$. Note that $b=\pm \lambda$ give the same $b^{-\frac{p^k+1}{2}}$ and one of $\pm \lambda$ is a square in $\mathbb{F}_{p^n}$ and the other is a nonsquare in $\mathbb{F}_{p^n}$. Hence, in order to determine $N(1,b)$ for $b\in \mathbb{F}_{p^n}^*$, we need to derive the mapping property of $h_i(x)=c$, $1\leq i\leq 4$, where $c$ is a square in $\mathbb{F}_{p^n}$, for $x\in E_{00}$ and $x\in E_{11}$, respectively. Then, using Lemma \ref{Lem_Map2}, the differential spectrum of $f(x)$ can be determined. 

Define the sets as
\begin{align*}
	\mathcal{H}_{ijk}=\{h_i(x)|x\in E_{jk}\}. 
\end{align*}

For $b\in C_{0}$, we should consider the mapping property of $h_i(x)=c$ on $E_{00}$, where $c$ is a square in $\mathbb{F}_{p^n}$. Assume that there exist $x_1,x_2\in E_{00}$ such that $h_1(x_1)=h_1(x_2)$ for $x_1\neq x_2$. Let $t_1=\theta^{-1}(x_1)$ and $t_2=\theta^{-1}(x_2)$. Then, from (\ref{Thm2_Eq1}), it is easy to derive that $(p^k+1)t_1\equiv \pm (p^k+1)t_2 \bmod p^n-1$. Since $(p^k+1,p^n-1)=2$, we have $t_1 \pm t_2 \equiv 0 \bmod (p^n-1)/2$, which cannot be satisfied because $1\leq t_1,t_2\leq (p^n-3)/4$. Hence we conclude that $h_1(x)|_{E_{00}}$ is injective on $E_{00}$, that is,  $|\mathcal{H}_{100}|=|E_{00}|=(p^n-3)/4$. 

Consider the mapping $h_2(x)|_{E_{00}}$, which has the same form as (\ref{Eq_Dfx_C00}). Therefore we can use the result when $p^n\equiv 3\bmod 4$ in Lemma \ref{Lem_MapProperty_00} and thus we have $|\mathcal{H}_{200}|=(p^n+p^k-2)/(2(p^k-1))$. The cardinality of the inverse image in $E_{00}$ of any element in $\mathcal{H}_{200}\setminus \{1\}$ is $(p^k-1)/2$ and  the cardinality of the inverse image in $E_{00}$ of $1\in \mathcal{H}_{200}$ is $(p^k-3)/4$.

Now, consider the relationship of the elements in $\mathcal{H}_{100}$ and $\mathcal{H}_{200}$. Assume that there exist $x_1,x_2\in E_{00}$ such that $h_1(x_1)=h_2(x_2)$. Let $t_1=\theta^{-1}(x_1)$ and $t_2=\theta^{-1}(x_2)$. Then, from (\ref{Thm2_Eq1}), we have $(p^k+1)t_1\equiv \pm (p^k-1)t_2 \bmod p^n-1$, which can be rewritten as
\begin{align}\label{Eq1_Thm2}
	\frac{p^k+1}{2}t_1\equiv \pm \frac{p^k-1}{2}t_2 \bmod \frac{p^n-1}{2}. 
\end{align}
Since $\gcd((p^k+1)/2,(p^n-1)/2)=1$, $(p^k+1)/2$ has an inverse modulo $(p^n-1)/2$. Hence for any $t_2\in \mathcal{T}_1$ which is not divisible by $(p^n-1)/(p^k-1)$, there exists $t_1\in \mathcal{T}_1$ satisfying (\ref{Eq1_Thm2}). Since $t_2$ which is divided by $(p^n-1)/(p^k-1)$ gives $h_2(x)=1$, we conclude that $1\in \mathcal{H}_{200}$ and $\mathcal{H}_{100}\supset (\mathcal{H}_{200}\setminus\{1\})$.

Since $h_4(x)=-h_1(x)$ and $h_3(x)=-h_2(x)$, $h_4(x)|_{E_{00}}$ has the same mapping property with $h_1(x)|_{E_{00}}$, and $h_3(x)|_{E_{00}}$ has the same mapping property with $h_2(x)|_{E_{00}}$. Furthermore, it is easy to check that $\mathcal{H}_{400}=-\mathcal{H}_{100}$, $\mathcal{H}_{300}=-\mathcal{H}_{200}$, and $\mathcal{H}_{400}\supset (\mathcal{H}_{300}\setminus\{-1\})$. 

It should also be checked that $\mathcal{H}_{100}$ cannot include both $y$ and $-y$. Assume that there exist $x_1,x_2\in E_{00}$ such that $h_1(x_1)=-h_1(x_2)$. From (\ref{Thm2_Eq1}), we have $(p^k+1)t_1\equiv \pm(p^k+1)t_2+(p^n-1)/2 \bmod p^n-1$, which can be rewritten as
\begin{align}\label{Eq2_Thm2}
	(p^k+1)(t_1\pm t_2)\equiv \frac{p^n-1}{2} \bmod p^n-1.
\end{align}
Since $\gcd(p^k+1,p^n-1)=2$ does not divide $(p^n-1)/2$, (\ref{Eq2_Thm2}) cannot be satisfied. Hence we conclude that there exist no $x_1,x_2\in E_{00}$ such that $h_1(x_1)=-h_1(x_2)$. Consequently, we conclude that $\mathcal{H}_{100}\cap \mathcal{H}_{400}=\emptyset$. 

So far, we have investigated the mapping property of $h_i(x)|_{E_{00}}$ and the relationship among the elements in $\mathcal{H}_{i00}$, $1\leq i\leq4$. 

Now, we will calculate that $N(1,b)=\lambda_1(b)+\lambda_2(b)+\lambda_3(b)+\lambda_4(b)$ for square $b\in 
\mathbb{F}_{p^n}^*$, which is the sum of the cardinalities of the inverse images in $E_{00}$ of the square element  in $\mathbb{F}_{p^n}$, $b^{-(p^k+1)/2}$ in (\ref{Eq_Thm2_add1}). Note that there are $(p^n-3)/4$ squares in $\mathcal{H}_{100}\cup\mathcal{H}_{400}$ because $\mathcal{H}_{100}\cap \mathcal{H}_{400}=\emptyset$ and $\mathcal{H}_{100}=-\mathcal{H}_{400}$. Since $\mathcal{H}_{100}\supset (\mathcal{H}_{200}\setminus \{1\})$, $\mathcal{H}_{400}\supset (\mathcal{H}_{300}\setminus \{-1\})$, and $\mathcal{H}_{200}=-\mathcal{H}_{300}$, there are $(p^n-p^k)/(2(p^k-1))$ squares in $(\mathcal{H}_{200}\setminus \{1\})\cup \mathcal{H}_{300}$, which are also included in $\mathcal{H}_{100}\cup \mathcal{H}_{400}$. We can regard each square in $\mathcal{H}_{100}\cup\mathcal{H}_{200}\cup\mathcal{H}_{300}\cup\mathcal{H}_{400}$ as $b^{-(p^k+1)/2}$ in (\ref{Eq_Thm2_add1}).
From Lemma \ref{Lem_Map2}, it is easy to check that for each square $c$ in $(\mathcal{H}_{200}\setminus \{1\})\cup \mathcal{H}_{300}$, $N(1,\delta)=(p^k+1)/2$ and for each square $c$ in $(\mathcal{H}_{100}\cup \mathcal{H}_{400})\setminus (\mathcal{H}_{200}\cup \mathcal{H}_{300}$), $N(1,\delta)=1$, where $\delta$ is a square in $\mathbb{F}_{p^n}$ such that $\delta^{-(p^k+1)/2}=c$. For $b=1$, from Lemma \ref{Lem_Map2}, $N(1,b)=(p^k-3)/4+1=(p^k+1)/4$. Let $n_i$ denote the number of $b\in \mathbb{F}_{p^n}$, which are squares in $\mathbb{F}_{p^n}$, such that  $N(1,b)=i$. Then, $n_{({p^k+1})/{2}}=(p^n-p^k)/(2(p^k-1))$, $n_{({p^k+1})/{4}}=1$, $n_{1}=(p^n-3)/4-(p^n-p^k)/(2(p^k-1))$, and $n_{0}=(p^n-1)/2-n_{({p^k+1})/{2}}-n_{1}$.

Consider the case when $b\in C_1$. From (\ref{Thm2_Eq1}) and (\ref{Thm2_Eq2}), note that $h_1(x)|_{E_{00}}=h_2|_{E_{11}}$, $h_4(x)|_{E_{00}}=h_4|_{E_{11}}$, $h_2(x)|_{E_{00}}=h_3|_{E_{11}}$, and $h_3(x)|_{E_{00}}=h_2|_{E_{11}}$. Since $t$ varies over $\mathcal{T}_1$ for both $x\in E_{00}$ and $x\in E_{11}$, they have the same mapping property, which means that for $b\in C_1$, the distribution of $N(1,b)$ is the same as the case when $b\in C_0$. Taking that $N(1,b)=1$ when $b=0$ into account, it is derived that  $\omega_{({p^k+1})/{2}}=2n_{({p^k+1})/{2}}=(p^n-p^k)/((p^k-1))$, $\omega_{({p^k+1})/{4}}=2n_{({p^k+1})/{4}}=2$,  $\omega_{1}=2n_1+1=(p^n-1)/2-(p^n-p^k)/((p^k-1))$, and $\omega_{0}=p^n-\omega_{({p^k+1})/{2}}-\omega_{({p^k+1})/{4}}-\omega_1$. 

\hfill $\Box$
\end{proof}

\begin{corollary}
\label{Mod_Helleseth}
For an odd prime $p$ such that $p\equiv 3\bmod 4$, odd $n$, $k|n$, and $d=(p^n+1)/(p^k+1)+(p^n-1)/2$, the differential uniformity of the function $f(x)=x^d$ in $\mathbb{F}_{p^n}$ is given as $\Delta_f=(p^k+1)/2$. \hfill $\Box$
\end{corollary}

From the results, new power functions which are differential $4$-uniform and $6$-uniform are  introduced as in the following corollaries. 

  \begin{corollary}
  \label{Cor_NewFunction}
  Let $d=({p^n+1})/{8}+({p^n-1})/{2}$. Then $x^d$ defined on $\mathbb{F}_{p^n}$ is differential $4$-uniform for $p=7$ and odd $n$.  \hfill $\Box$
  \end{corollary}  
  
   \begin{corollary}
  \label{Cor_NewFunction2}
  Let $d=({p^n+1})/{12}+({p^n-1})/{2}$. Then $x^d$ defined on $\mathbb{F}_{p^n}$ is differential $6$-uniform for $p=11$ and odd $n$.  \hfill $\Box$
  \end{corollary}

\section{Conclusion}
\label{Concl}

In this paper, the differential spectrum of the two power functions $x^{\frac{p^k+1}{2}}$ and $x^{\frac{p^n+1}{p^k+1}+\frac{p^n-1}{2}}$ in $\mathbb{F}_{p^n}$ are derived. The result can be used to determine the differential uniformity $\Delta_f$ of the two power functions. Two new power functions in $\mathbb{F}_{p^n}$ which are differential $4$-uniform and $6$-uniform  are also found. 

\section*{Appendix}
\begin{aproof}
{
~~~\\

\noindent {\bf Case 1)} Relationship between $\mathcal{I}_{00}$ and $\mathcal{I}_{10}$; 
\vspace{0.2cm}

Assume that there exist $x_1\in E_{00}$ and $x_2\in E_{10}$ such that $D_f(x_1)=D_f(x_2)$. Let $t_1=\theta^{-1}(x_1)$ and $t_2=\theta^{-1}(x_2)$. From Lemma \ref{Eq_C00}, we have
\begin{align}\label{Eq_Rela2_1}
	\alpha^{(p^k-1)t_1}+\alpha^{-(p^k-1)t_1}=\delta^{2(p^k-1)t_2}+\delta^{-2(p^k-1)t_2}.
\end{align}
Since $\alpha=\beta^{p^n+1}$ and $\delta=\beta^{(p^n-1)/2}$, (\ref{Eq_Rela2_1}) is satisfied if and only if 
\begin{align}\label{0521_A}
	[(p^n+1)t_1\pm (p^n-1)t_2](p^k-1)\equiv 0 \bmod (p^{2n}-1).
\end{align}
Then (\ref{0521_A}) can be rewritten as
\begin{align}\label{0521_B}
	(p^n+1)t_1\pm (p^n-1)t_2\equiv 0 \bmod \frac{p^{2n}-1}{p^g-1}
\end{align}
where $\gcd(p^k-1,p^{2n}-1)=p^g-1$.

Note that
\begin{align}\label{Eq2_ConB}
	\gcd(\frac{p^{2n}-1}{p^g-1},p^n+1)=
	\begin{cases}
		p^n+1,&~{\rm if~}g=e\\
		\frac{p^n+1}{p^e+1},&~{\rm if~}g=2e.
	\end{cases}
\end{align}

Consider the case when $g=e$, i.e., $k/e$ is odd. For the solvability of (\ref{0521_A}), $\pm (p^n-1)t_2$ should be divided by $p^n+1$. Since $\gcd(p^n+1,p^n-1)=2$, $t_2$ should be divided by $(p^n+1)/2$. Since $t_2$ varies over $[1,(p^n-3)/4]$ for $p^n\equiv 3\bmod 4$ and $[1,(p^n-1)/4]$ for $p^n\equiv 1\bmod 4$, $t_2$ cannot be divided by $(p^n+1)/2$.  Hence we conclude that $\mathcal{I}_{00}\cap \mathcal{I}_{10}=\emptyset$ for odd $k/e$. 

Next, consider the case when $g=2e$, i.e., $k/e$ is even. From (\ref{0521_B}), $(p^n+1)t_1$ should be divided by $\gcd(p^n-1,(p^{2n}-1)/(p^{2e}-1))=(p^n-1)/(p^e-1)$. Since $\gcd((p^n-1)/(p^e-1),p^n+1)=1$, $t_1$ should be divided by $(p^n-1)/(p^e-1)$, which means that $\mathcal{I}_{00}\cap \mathcal{I}_{10}=\{1\}$ for even $k/e$. 
  
\vspace{0.1cm}  
\noindent {\bf Case 2)} Relationship between $\mathcal{I}_{11}$ and $\mathcal{I}_{10}$;
\vspace{0.2cm}

For the case when $p^n\equiv 3\bmod 4$ and $p^k\equiv 1\bmod 4$, we already proved that $\mathcal{I}_{00}=\mathcal{I}_{11}$. Since $g=2e$, i.e., $k/e$ is even, in the case, we conclude that $\mathcal{I}_{11}\cap \mathcal{I}_{10}=\{1\}$ for even $k/e$.

Consider the case when $p^n\equiv 3\bmod 4$ and $p^k\equiv 3\bmod 4$. Note that $g=e$, i.e., $k/e$ is odd in the case. We can prove this case similar to Case 1). Assume that there exist $x_1\in E_{11}$ and $x_2\in E_{10}$ such that $D_f(x_1)=D_f(x_2)$. Let $t_1=\theta^{-1}(x_1)$ and $t_2=\theta^{-1}(x_2)$. From Lemma \ref{Eq_C00}, we have
\begin{align}\label{Eq_Rela2_222}
	(p^k-1)[(p^n+1)t_1\pm (p^n-1)t_2] \equiv \frac{p^{2n}-1}{2} \bmod (p^{2n}-1).
\end{align}
Then (\ref{Eq_Rela2_222}) can be rewritten as
\begin{align}\label{Eq1_ConC}
	(p^n+1)t_1\pm (p^n-1)t_2 \equiv \frac{p^{2n}-1}{2(p^e-1)} \bmod \frac{p^{2n}-1}{p^e-1}
\end{align}
where $\gcd(p^k-1,(p^{2n}-1)/2)=p^e-1$. For the solvability of (\ref{Eq1_ConC}), $\pm (p^n-1)t_2$ should be divided by $(p^n+1)/2$. Since $\gcd((p^n+1)/2,p^n-1)=2$, $t_2$ should be divided by $(p^n+1)/4$. Since $t_2$ varies over $[1,(p^n-3)/4]$ for $p^n\equiv 3\bmod 4$, $t_2$ cannot be divided by $(p^n+1)/4$.  Hence we conclude that $\mathcal{I}_{00}\cap \mathcal{I}_{10}=\emptyset$ for odd $k/e$.

Next, consider the case when $p^n\equiv 1\bmod 4$ and $p^k\equiv 1\bmod 4$. Assume that there exist $x_1\in E_{11}$ and $x_2\in E_{10}$ such that $D_f(x_1)=D_f(x_2)$. From Lemma \ref{Eq_C00} and by setting $\gamma=-\alpha$,  we have
\begin{align}\label{Eq_Con3_3}
	(p^k-1)[(p^n+1)t_1\pm (p^n-1)t_2] \equiv -\frac{p^k-1}{2}(p^n+1) \bmod (p^{2n}-1).
\end{align}
Note that $g$ can be equal to either $e$ or $2e$ in this case. For the case when $g=e$, i.e., $k/e$ is odd, (\ref{Eq_Con3_3}) can be rewritten as
\begin{align}\label{Eq_Con3_4}
	\frac{p^k-1}{p^e-1}[(p^n+1)t_1\pm (p^n-1)t_2] \equiv -\frac{p^n+1}{2}\cdot\frac{p^k-1}{p^e-1} \bmod \frac{p^{2n}-1}{p^e-1}.
\end{align}
From (\ref{Eq_Con3_4}), $(p^n-1)t_2$ should be divided by $(p^n+1)/2$. Since $\gcd(p^n-1,(p^n+1)/2)=1$, $t_2$ should be divided by $(p^n+1)/2$. Note that $t_2$ varies over $[1,(p^n-1)/4]$. We conclude that $\mathcal{I}_{11}\cap \mathcal{I}_{10}=\emptyset$ for odd $k/e$. 

For the case when $g=2e$, i.e., $k/e$ is even, (\ref{Eq_Con3_3}) can be rewritten as
\begin{align}\label{Eq_RelC_1}
	\frac{p^k-1}{p^g-1}[(p^n+1)t_1\pm (p^n-1)t_2] \equiv -\frac{p^n+1}{2}\cdot\frac{p^k-1}{p^g-1} \bmod \frac{p^{2n}-1}{p^g-1}.
\end{align} 
From $\gcd((p^{2n}-1)/(p^g-1),(p^n+1)/2)=(p^n+1)/(p^e+1)$ and (\ref{Eq_RelC_1}), $(p^n-1)t_2$ should be divided by $(p^n+1)/(p^e+1)$. Since $\gcd(p^n-1,(p^n+1)/(p^e+1))=1$, $t_2$ should be divided by $(p^n+1)/(p^e+1)$. From Lemma \ref{Lem_gcd}, we have $p^e+1|p^k-1$. Hence $t_2$ which is divided by $(p^n+1)/(p^e+1)$ gives $D_f(x)=1$, which means that $\mathcal{I}_{11}\cap \mathcal{I}_{10}=\{1\}$ for even $k/e$. 

The case when $p^n\equiv 1\bmod 4$ and $p^k\equiv 3\bmod 4$ can be proved similarly.

\vspace{0.1cm}
\noindent {\bf Case 3)} Relationship between $\mathcal{I}_{00}$ and $\mathcal{I}_{01}$;
\vspace{0.2cm}

We already proved that $\mathcal{I}_{10}\cap \mathcal{I}_{01}=\emptyset$ for $p^e\equiv 3\bmod  4$ and odd $k/e$ and  $\mathcal{I}_{10}= \mathcal{I}_{01}$, otherwise. Hence we only need to consider the case when $p^e\equiv 3\bmod  4$ and odd $k/e$ in Case 3) and Case 4). Note that $p^k\equiv 3\bmod 4$ in this case. 

First, consider the relationship between $\mathcal{I}_{00}$ and $\mathcal{I}_{01}$. Assume that there exist $x_1\in E_{00}$ and $x_2\in E_{01}$ such that $D_f(x_1)=D_f(x_2)$. Let $t_1=\theta^{-1}(x_1)$ and $t_2=\theta^{-1}(x_2)$. Again, we have 
\begin{align}\label{Rel4_Eq_1}
	(p^k-1)[(p^n+1)t_1\pm (\frac{p^n-1}{2}+(p^n-1)t_2)]\equiv 0\bmod (p^{2n}-1),
\end{align}
which can be rewritten as
\begin{align}\label{Rel4_Eq_2}
	\frac{p^k-1}{p^e-1}[(p^n+1)t_1\pm (\frac{p^n-1}{2}+(p^n-1)t_2)]\equiv 0 \bmod \frac{p^{2n}-1}{p^e-1}.
\end{align}
For the solvability of (\ref{Rel4_Eq_2}), $(p^n-1)/2\pm(p^n-1)t_2$ should be divided by $p^n+1$, which is given as
\begin{align}\label{Eq_0417}
	\frac{p^n-1}{2}(1\pm 2t_2)\equiv 0 \bmod (p^n+1).
\end{align}
For $p^n\equiv 3\bmod 4$, the left-hand side is odd, while the right-hand side is even, which is a contradiction. For $p^n\equiv 1\bmod 4$, since $\gcd((p^n-1)/2,p^n+1)=2$, $1\pm 2t_2$ should be divided by $(p^n+1)/2$. Assume that $1+2t_2=(p^n+1)/2$. Then $t_2$ should be $(p^n-1)/4$. However, since $t_2$ varies over $[0,(p^n-5)/4]$, it is impossible. Therefore, we conclude that $\mathcal{I}_{00}\cap \mathcal{I}_{01}=\emptyset$. 

\vspace{0.1cm}
\noindent {\bf Case 4)} Relationship between $\mathcal{I}_{11}$, and $\mathcal{I}_{01}$;
\vspace{0.2cm}

Next, consider the relationship between $\mathcal{I}_{11}$ and $\mathcal{I}_{01}$. Assume that there exist $x_1\in E_{11}$ and $x_2\in E_{01}$ such that $D_f(x_1)=D_f(x_2)$. For the case when $p^n\equiv 3\bmod 4$, by setting $\gamma=-1$, we have
\begin{align}\label{Rel4_eq_1}
	(p^k-1)[(p^n+1)t_1\pm (\frac{p^n-1}{2}+(p^n-1)t_2)]\equiv \frac{p^{2n}-1}{2}\bmod (p^{2n}-1),
\end{align}
which can be rewritten as
\begin{align}\label{Rel4_eq_2}
	\frac{p^k-1}{p^e-1}[(p^n+1)t_1\pm (\frac{p^n-1}{2}+(p^n-1)t_2)]
\equiv \frac{p^n+1}{2}\cdot\frac{p^n-1}{p^e-1} \bmod \frac{p^{2n}-1}{p^e-1}.
\end{align}
For the solvability of  (\ref{Rel4_eq_2}), $(p^n-1)\big(1\pm 2(p^n-1)t_2\big)/2$ should be divided by $(p^n+1)/2$. Since $\gcd((p^n-1)/2,(p^n+1)/2)=1$, $1\pm 2(p^n-1)t_2$ should be divided by $(p^n+1)/2$. Since $(p^n+1)/2$ is even and $1\pm 2(p^n-1)t_2$ is odd, it is a contradiction. Hence  we conclude that $\mathcal{I}_{00}\cap \mathcal{I}_{01}=\emptyset$. 

For the case when $p^n\equiv 1\bmod 4$, $(p^n-1)/2(1\pm 2t_2)$ should be divided by $(p^n+1)/2$. Hence $1\pm 2t_2$ should be divided by $(p^n+1)/2$. Assume that $1+2t_2$ is divided by $(p^n+1)/2$. Then $t_2$ should be equal to $(p^n-1)/4$, which is a contradiction because $t_2\leq (p^n-5)/4$. Therefore, $\mathcal{I}_{00}\cap \mathcal{I}_{01}=\emptyset$.

From Case 1)--Case 4), the proof can be done. \hfill$\Box$
}
\end{aproof}


\bibliographystyle{elsarticle-num}
\bibliography{<your-bib-database>}






% \end{thebibliography}

\end{document}